\documentclass[12pt]{amsart}

\usepackage{amsthm}
\usepackage{amsmath}
\usepackage{amsfonts}
\usepackage{amssymb}
\usepackage{mathdots}

\usepackage{theoremref}
\newtheorem{Thm}{Theorem}[section]

\newtheorem{Pro}[Thm]{Proposition}

\theoremstyle{definition}

\theoremstyle{remark}
\newtheorem{Rmk}[Thm]{Remark}

\DeclareMathOperator\arccosh{arccosh}

\renewcommand{\d}{\partial}
\newcommand{\eps}{\varepsilon}

\begin{document}
\title{Stochastic differential equations for Lie group valued moment maps}
\author{Anton Alekseev\and Elizaveta Arzhakova\and Daria Smirnova}

\newcommand{\Addresses}{{
  \bigskip
  \footnotesize

  \textsc{AA: Section of Mathematics, University of Geneva, 2-4 rue du Li\`evre, c.p. 64, 1211 Gen\`eve 4, Switzerland}\par\nopagebreak
  \textit{E-mail address}: \texttt{Anton.Alekseev@unige.ch}

  \medskip  
  
  \textsc{EA: Mathematical Institute, University of Leiden, P.O. Box 9512, 2300 RA Leiden, The Netherlands}\par\nopagebreak
  \textit{E-mail address}: \texttt{e.arzhakova@math.leidenuniv.nl}
  
  \medskip

  \textsc{DS: Section of Mathematics, University of Geneva, 2-4 rue du Li\`evre, c.p. 64, 1211 Gen\`eve 4, Switzerland}\par\nopagebreak
  \textit{E-mail address}: \texttt{daria.smirnova@unige.ch}

}}
\date{}

\begin{abstract}
The celebrated result by Biane-Bougerol-O'Connell relates  Duistermaat-Heckman (DH) measures for coadjoint orbits of a compact Lie group $G$ with the multi-dimensional Pitman transform of the Wiener process on its Cartan subalgebra. The DH  theory admits several non-trivial generalizations. In this paper, we consider the case of $G=SU(2)$, and we give an interpretation of DH measures for $SU(2) \cong S^3$ valued moment maps in terms of an interesting stochastic process on the unit disc, and an interpretation of the DH measures for Poisson $\mathbb{H}^3$ valued moment maps (in the sense of Lu) in terms of a stochastic process on the interior of a hyperbola.

\end{abstract}

\maketitle


\section{Introduction: the Wiener process on $\mathbb{R}^3$ and Duistermaat-Heckman measures of $S^2$}

The vector space $\mathbb{R}^3$ carries two amazing geometric structures. On the one hand, it has a Euclidean metric 
$$
g=dx^2 + dy^2 + dz^2.
$$
On the other hand, it is isomorphic to the dual space of the Lie algebra $SU(2)$ and it carries the linear Kirillov-Kostant-Souriau (KKS) Poisson structure \cite{Kir99, Sou70}
$$
\{ x, y\} = z, \hskip 0.3cm \{ y, z\} = x, \hskip 0.3cm \{ z, x\} =y.
$$
{\em A priori}, these two structures are unrelated, but there is an interesting connection between them that we will now describe.

First, consider the Poisson structure on $\mathbb{R}^3$. It gives rise to a generalized foliation by symplectic leaves. In this case, the leaves are 2-spheres centered at the origin. They carry symplectic forms which are given by formula
$$
\omega = d\phi \wedge dz,
$$
where $x+iy=\sqrt{r^2-z^2} e^{i\phi}$ and $r$ is the radius of the sphere. The 2-form $\omega$ is a rotation invariant volume form on the sphere and its push-forward to the $z$-axis is called the Duistermaat-Heckman measure \cite{DH82}:
$$
{\rm DH}_r= 2\pi \chi_{[-r, r]}(z) dz.
$$
Here $\chi_{[-r,r]}(z)$ is the characteristic function of the segment $[-r,r]$. The mass of this measure is equal to the symplectic volume of the sphere given by ${\rm Vol}(S^2, \omega)=4\pi r$. The normalized measure
$$
\frac{1}{{\rm Vol}(S^2, \omega)} {\rm DH}_r= \frac{1}{2r} \chi_{[-r, r]}(z) dz
$$
is a probability measure.

Next, use the Euclidean metric on $\mathbb{R}^3$ to define a Wiener process which starts at the origin. We will consider two projections of this process: the first one under the map 
$$
r: (x,y,z) \mapsto r=\sqrt{x^2+y^2+z^2}
$$
and the second one under the map $(x,y,z) \mapsto (r,z)$. These two projections are described by the following properties:

\begin{Thm}  \label{thm:Bessel_R^3}
The projection $r$ of the Wiener process on $\mathbb{R}^3$ is a Markov process described by the stochastic differential equation
\begin{equation}  \label{eq:Bessel_R^3}
dr_t = dB_t+ \frac{1}{r_t} dt,
\end{equation}
where $B_t$ is the standard Wiener process on $\mathbb{R}$.
\end{Thm}

\begin{Rmk}
For the proof of Theorem \ref{thm:Bessel_R^3} see {\em e.g.} \cite{McK60, RY91}. The process \eqref{eq:Bessel_R^3} is the 3-dimensional Bessel process.
\end{Rmk}

\begin{Thm}
The projection $(r,z)$ of the Wiener process on $\mathbb{R}^3$ is a Markov process described by the following system of stochastic differential equations
\begin{equation}  \label{system_R^3}
\begin{array}{lll}
dr_t & = & \frac{\sqrt{r_t^2-z_t^2}}{r_t} dB_t^{(1)} + \frac{z_t}{r_t} dB_t^{(2)} + \frac{1}{r_t} dt,\\
dz_t & = & dB^{(2)}_t,
\end{array}
\end{equation}
where $B^{(1,2)}_t$ are two independent Wiener processes on $\mathbb{R}$.
\end{Thm}
The following statement establishes a surprizing relation between the system of stochastic differential equation \eqref{system_R^3} and the Duistermaat-Heckman measure:

\begin{Thm} \label{conditional_R^3}
The conditional probability density for $z_t$ for a fixed value of $r_t$ is given by the normalized Duistermaat-Heckman measure:
$$
\rho_{z_t} (z| r_t=r) = \frac{1}{{\rm Vol}(S^2, \omega)} {\rm DH}_r.
$$
\end{Thm}
We do not prove these two theorems here since we give proofs of similar (and somewhat more involved) statements in the body of the paper. The purpose of this paper is to generalize the above results to the cases of spherical and hyperbolic geometry in dimension 3. In both cases, there is a Riemannian metric which gives rise to a well-defined Wiener process.
The Duistermaat-Heckman measure also admits generalizations and analogues of Theorem \ref{conditional_R^3} hold true.

Our work is inspired by another way to relate Wiener processes to the Duistermaat-Heckman measure \cite{BBC08, Bia09}. In more detail,
$$
\rho_{B_t}(z|r_t=PB_t) = \frac{1}{{\rm Vol}(S^2, \omega)} {\rm DH}_r,
$$
where $PB_t$ is the Pitman transform (for definition, see \cite{Pit75}) of the Wiener process on $\mathbb{R}$. It would be desirable to have similar results in the case of spherical and hyperbolic geometry, but to the best of our knowledge they are currently not available.

In the process of completing this paper, we learnt that related results were obtained in \cite{CC}.

\vskip 0.2cm

{\bf Acknowledgements.} We are grateful to D. Chelkak and L. Parnovksi for useful comments and suggestions.
The research of AA and DS was supported in part by the grants 178794, 178828, 182767  and by the NCCR SwissMAP of the Swiss National Science Foundation. The research of AA was supported in part by the project MODFLAT of the European Research Council (ERC). 
This work is partly based on the Master thesis of EA at the University of Geneva.

\section{Wiener process on $S^3$ and group valued moment maps}

In this section, we consider the 3-sphere $S^3$ which replaces the Euclidean space $\mathbb{R}^3$ of the previous section. It is convenient to identify $S^3$ with the Lie group 
$$
SU(2)=\left\{
g= \left(  
\begin{array}{rr}
a& b \\
- \bar{b} & \bar{a}
\end{array}
\right); |a|^2 + |b|^2 =1 \right\} .
$$
It carries a quasi-Poisson structure (see \cite{AKSM02}) with leaves the conjugacy classes. Conjugacy classes in $SU(2)$ are the points 
$e=\left(
\begin{array}{rr}
1 & 0 \\
0 & 1
\end{array}
\right)$
and $-e$ and spheres given by matrices of fixed trace. Consider the map $a: g \mapsto a$ which picks the left upper corner matrix element of $g$. It is convenient to introduce a Cartesian and polar coordinate systems for 
$$
a=x+iy = \rho e^{i\phi}.
$$
If we denote by $\lambda=e^{i\theta}$ the eigenvalue of $g$ with non-negative imaginary part, we have
$$
x = \rho \cos(\phi) = \cos(\theta).
$$
According to \cite{AMW02}, conjugacy classes in $G=SU(2)$ carry canonical volume forms $\omega_\theta$. The following statement replaces the formula for the Duistermaat-Heckman measure:
\begin{Pro}
Let $\mathcal{C}_\theta = \{ g \in SU(2); {\rm Tr}(g)=2 \cos(\theta)\} $ be a conjugacy class in $SU(2)$. Then,
$$
a(\mathcal{C}_\theta) =\{ \cos{\theta} + iy; y\in [-\sin(\theta), \sin(\theta)] \} .
$$
Furthermore, 
\begin{equation} \label{DH_SU(2)}
{\rm DH}_\theta:= a_*(\omega_\theta) = 2\pi \chi_{[-\sin(\theta), \sin(\theta)]} dy
\end{equation}
\end{Pro}

The volume of the conjugacy class $\mathcal{C}_\theta$ is given by $4\pi \sin(\theta)$, and it gives the total mass of the measure \eqref{DH_SU(2)}. The normalized measure
$$
\frac{1}{{\rm Vol}(\mathcal{C}_\theta)} {\rm DH}_\theta = \frac{1}{2 \sin(\theta)} \chi_{[-\sin(\theta), \sin(\theta)]} dy
$$
is a probability measure.

The space $SU(2) \cong S^3$ has the unique (up to multiple) bi-invariant metric. Consider the Wiener process under this metric which starts at the group unit $e$. We will again consider two projections $\theta: SU(2) \to [0, \pi]$ and $a: SU(2) \to D \subset \mathbb{C}$, where $D$ is the unit disc. These projections have the following properties:

\begin{Thm} \label{SU2_cot}
The projection $\theta$ of the Wiener process on $S^3$ is a Markov process described by the stochastic differential equation
$$
d\theta_t = dB_t+ \cot(\theta_t) dt,
$$
where $B_t$ is the standard Wiener process on $\mathbb{R}$.
\end{Thm}

\begin{Thm}
The projection $a=x+iy$ of the Wiener process on $S^3$ is a Markov process described by the following system of stochastic differential equations
\begin{equation}  \label{system_S^3}
\begin{array}{lll}
dx_t & = & \frac{y_t}{\sqrt{x_t^2+y_t^2}} dB_t^{(1)} + \frac{x_t\sqrt{1-x_t^2-y_t^2}}{\sqrt{x_t^2+y_t^2}} dB_t^{(2)} - \frac{3}{2} x_t dt,\\
dy_t & = & -\frac{x_t}{\sqrt{x_t^2+y_t^2}} dB_t^{(1)} + \frac{y_t\sqrt{1-x_t^2-y_t^2}}{\sqrt{x_t^2+y_t^2}} dB_t^{(2)} - \frac{3}{2} y_t dt,
\end{array}
\end{equation}
where $B^{(1,2)}_t$ are two independent Wiener processes on $\mathbb{R}$.
\end{Thm}

\begin{Rmk}
In polar coordinates $\rho=\sqrt{x^2+y^2}, \phi=\arctan(y/x)$ the system of stochastic differential equations \eqref{system_S^3} acquires a beautiful form:

\begin{equation}
\begin{array} {ll}
d\rho &= \sqrt{1-\rho^2}d \tilde{B}^{(1)}_t + \frac{1-3\rho^2}{2\rho}dt, \\
d\phi &= \frac{1}{\rho} d \tilde{B}^{(2)}_t,
\end{array}
\end{equation}

where $d \tilde{B}^{(1)}_t$ and $d \tilde{B}^{(2)}_t$ are independent Wiener processes.  
\end{Rmk}

\begin{proof}
The Wiener process $B_t^{S^3}$ on $S^3 \cong SU(2)$  is described by the following matrix equation:
\begin{equation}   \label{eq:BM_S}
dg = g \left( \sum_{i=1}^3 e_i dB^{(e_i)}_t \right) -\frac{3}{2} g_t dt.    
\end{equation}
Here $\{ e_i\}$ are orthonormal generators of the Lie algebra 
$\mathfrak{su}(2)$:
\begin{equation} \nonumber
e_1= \left(  
\begin{array}{rr}
0   &   i \\
i   &   0
\end{array}
\right), \quad
e_2= \left(  
\begin{array}{rr}
0   &   -1 \\
1   &   0
\end{array}
\right), \quad
e_3= \left(  
\begin{array}{rr}
i   &   0 \\
0   &   -i
\end{array}
\right),
\end{equation}
and $B^{(e_i)}$ are independent Wiener processes. The drift term is introduced to preserve the determinant of $g$. 
Equation \eqref{eq:BM_S} implies the following stochastic differential equations for $a_t$ and $b_t$:
\begin{equation} \label{system_S^3_ab}
\begin{array}{lll}
d a_t &=& b_t (i d B^{(e_1)}_t + d B^{(e_2)}_t) + i a_t d B_t^{(e_3)} - \frac 3 2 a_t dt, \\
d b_t &=& a_t (i d B^{(e_1)}_t - d B^{(e_2)}_t) - i b_t d B_t^{(e_3)} - \frac 3 2 b_t dt. \\
\end{array}
\end{equation}
Then, the evolution of $x_t = Re(a_t)$ is given by:
\begin{equation} \label{system_S^3_x}
\begin{array}{lll}
d x_t &=& - Im(b_t) d B^{(e_1)}_t + Re(b_t) d B^{(e_2)}_t - Im(a_t) d B^{(e_3)}_t - \frac{3}{2} x_t dt \\
&=& \sqrt{1 - x^2_t}d B_t - \frac 3 2 x_t dt,  \\
\end{array}
\end{equation}
where $B_t$ is the standard Wiener process on $\mathbb{R}$. The second line is the stochastic process whose distribution is equal to the one of the first line. Applying It\^{o}'s Lemma to the second equation of \eqref{system_S^3_x} with $\theta = \arccos(x)$, we obtain the statement of Theorem \ref{SU2_cot}.

The system of stochastic differential equations for $x_t$ and $y_t = Im(a_t)$ given by \eqref{system_S^3} also follows from \eqref{system_S^3_ab}.
Note that the correlation matrix of the processes $x_t$ and $y_t$ is equal to
$$
\left(
\begin{array} {ll}
1 - x_t^2 & x_t y_t \\
x_t y_t & 1 - y_t^2
\end{array}
\right) ,
$$
and this defines the coefficients in front of normalized independent Wiener processes $B_t^{(1)}$ and $B_t^{(2)}$.
\end{proof}


The following theorem establishes a relation between the system of stochastic differential equation \eqref{system_S^3} and Duistermaat-Heckman measures of conjugacy classes:

\begin{Thm}
 \label{conditional_S^3}
The conditional probability density for $y_t$ for a fixed value of $x_t$ is given by the normalized Duistermaat-Heckman measure:
\begin{equation} \label{conditional_S^3_eq}
\rho_{y_t} (y| x_t=\cos(\theta)) = \frac{1}{{\rm Vol}(\mathcal{C}_\theta)} {\rm DH}_\theta.
\end{equation}
\end{Thm}

\begin{proof} [Proof of Theorem \ref{conditional_S^3}]
We compare two Fokker-Plank equations 
on evolution of the probability densities for $x_t$ and for the combined process  $(x_t, y_t)$.
The Fokker-Planck equation derived from \eqref{system_S^3_x} reads
\begin{equation} \label{FP1_S3}
\frac{d}{d t} p_x 
= \frac{1}{2} (1 - x^2) \frac{\d^2}{\d x^2} p_x 
- \frac{1}{2} x \frac{\d}{\d x} p_x + \frac{1}{2} p_x
\end{equation}
for $p_x = p_x(x,t)$ the probability density of $x_t$.

The Fokker-Planck equation describing the distribution $p_{x,y} = p_{x,y}(x,y,t)$ for for the process $(x_t, y_t)$ and derived from \eqref{system_S^3} is as follows:
\begin{equation} \label{FP2_S3}
\begin{array}{lll}
\frac{d}{d t} p_{x,y} 
&=& 
\frac 1 2 (1-x^2) \frac{\d^2}{\d x^2} p_{x,y} + 
\frac 1 2 (1-y^2) \frac{\d^2}{\d y^2} p_{x,y} - xy \frac{\d^2}{\d x \d y} p_{x,y} \\
&& - \frac 3 2 x \frac{\d}{\d x} p_{x,y} -\frac 3 2 y \frac{\d}{\d y} p_{x,y}. 
\end{array}
\end{equation}

Equations \eqref{FP1_S3} and \eqref{FP2_S3} coincide if $p_{x,y}$ is of the form 
\begin{equation} \label{FPanz_S3}
p_{x,y}(x,y,t) = p_x(x,t) \times \frac{1}{2 \sqrt{1 - x^2}} \chi_{[-\sqrt{1 - x^2},\sqrt{1 - x^2}]}(y).
\end{equation}
Hence, if the initial conditions are of this form, the solution $p_x(x,t)$ of \eqref{FP1_S3} yields a solution of \eqref{FP2_S3} via \eqref{FPanz_S3}. Consider the initial condition
$$p_{0}(x,y) = \delta_{\sqrt{1- \eps^2}}(x) \times \tfrac{1}{2 \eps} \chi_{[-\eps,\eps]}(y)$$
which is of the form \eqref{FPanz_S3}. Then, $p_{x,y}(x,y,t)$ is also of the form \eqref{FPanz_S3}. It holds when $\eps \to 0$ as well, and this implies equation \eqref{conditional_S^3_eq} for the conditional probability for the projection of the Wiener process on $S^3$ starting at the group unit.

\end{proof}

\section{Wiener process on $\mathbb{H}^3$ and moment maps in the sense of Lu}

Similar to the previous sections, we now replace the Euclidean space $\mathbb{R}^3$ with the hyperbolic space $\mathbb{H}^3$.
A good model of $\mathbb{H}^3$ is the set of positive definite Hermitian matrices of unit determinant:

$$
\mathbb{H}^3 \cong \left\{
g= \left(  
\begin{array}{rr}
a& b \\
\bar{b} & c
\end{array}
\right); a,c \in \mathbb{R}_+, b \in \mathbb{C}, ac - |b|^2 =1 \right\} .
$$
It carries a Lu-Weinstein Poisson structure \cite{LW90} and a quasi-Poisson structure \cite{AKSM02}. The leaves for both structures are conjugacy classes under the $SU(2)$-action. Generic leaves are 2-spheres of elements of $g \in \mathbb{H}^3$ with fixed trace. The conjugacy class of the unit matrix $e$ consists of one point.

A matrix $g \in \mathbb{H}^2$ has positive eigenvalues $\Lambda, \Lambda^{-1}$ with $\Lambda \geq 1$. We denote $\lambda = \log(\Lambda)$ and (by abusing notation) we denote $\lambda: \mathbb{H}^3 \to \mathbb{R}$ the corresponding map. We also consider the maps
$a,c: \mathbb{H}^3 \to \mathbb{R}$ mapping an element $g$ to its diagonal entries $a$ and $c$.

Again, $SU(2)$ conjugacy classes in $\mathbb{H}^3$ carry canonical volume forms corresponding to Lu-Weinstein Poisson structures and to quasi-Poisson structures. We denote the volume form corresponding to the quasi-Poissons structure by $\omega_\lambda$. 
The analogue of the Duistermaat-Heckman measure is given by:
\begin{Pro}
Let $\mathcal{C}_\lambda = \{ g \in \mathbb{H}^3; {\rm Tr}(g)=e^\lambda + e^{-\lambda}\} $ be a leaf in $\mathbb{H}^3$. Then,
$$
(a,c)(\mathcal{C}_\lambda) =\{ (a, 2\cosh(\lambda)-a); a\in [e^{-\lambda}, e^\lambda] \} .
$$
Furthermore, 
\begin{equation} \label{DH_H^3}
{\rm DH}_\lambda:= a_*(\omega_\lambda) = 2\pi \chi_{[e^{-\lambda}, e^\lambda]} da
\end{equation}
\end{Pro}

The corresponding normalized measure is of the form
$$
\frac{1}{{\rm Vol}(\mathcal{C}_\lambda)} {\rm DH}_\lambda = \frac{1}{2 \sinh(\lambda)} \chi_{[e^{-\lambda}, e^\lambda]} da .
$$

The space $\mathbb{H}^3$ carries the standard hyperbolic metric 
$$
d(g_1, g_2) = \arccosh(\tfrac{1}{2}\mathrm{tr}(g_1 g^*_2)).
$$ 
and we consider the Wiener process on $\mathbb{H}^3$ under this metric which starts at the unit element $e$. One can write it explicitly in local coordinates, {\it{e.g.}} given in \cite{Cos01}.

%


As before, we will be comparing two projections of this 3-dimensional Wiener process, the first one is under the (logarithmic) eigenvalue map $\lambda: \mathbb{H}^3 \to \mathbb{R}$ and the second one is under the map $(\tfrac{a+c}{2} = w,c): \mathbb{H}^3 \to \mathbb{R}^2$. These projections have the following properties:

\begin{Thm} \label{H^3_coth}
The projection $\lambda$ of the Wiener process on $\mathbb{H}^3$ is a Markov process described by the stochastic differential equation
$$
d\lambda_t = dB_t+ \coth(\lambda_t) dt,
$$
where $B_t$ is the standard Wiener process on $\mathbb{R}$.
\end{Thm}

\begin{Thm} \label{H^3_cw}
The projection $(w,c)$ of the Wiener process on $\mathbb{H}^3$ is a Markov process described by the following system of stochastic differential equations
\begin{equation}  \label{system_H^3}
\begin{array}{lll}
dw_t & = & \sqrt{w_t^2-1} \, dB_t^{(1)} + 
\frac{3}{2} w_t dt,\\
dc_t & = &
\frac{c_t w_t -1}{\sqrt{w_t^2-1}} dB_t^{(1)} +
\frac{\sqrt{2 c_t w_t -c_t^2 - 1}}{\sqrt{w_t^2-1}} dB_t^{(2)} + \frac{3}{2} c_t dt,
\end{array}
\end{equation}
where $B^{(1,2)}_t$ are two independent Brownian motions on $\mathbb{R}$.
\end{Thm}

Similar to Theorem \ref{system_S^3}, Theorem \ref{H^3_cw} follows directly from stochastic differential equations for the Wiener process on $\mathbb{H}^3$.
To obtain Theorem \ref{H^3_coth}, we apply It\^{o}'s lemma with $\lambda = \arccosh(w)$ to the first equation of \eqref{system_H^3}.

The following theorem establishes a relation between the system of stochastic differential equation \eqref{system_H^3} and Duistermaat-Heckman measures of conjugacy classes:

\begin{Thm} \label{conditional_H^3}
The conditional probability density for $c_t$ for a fixed value of $w_t=\cosh(\lambda_t)$ is given by the normalized Duistermaat-Heckman measure:
\begin{equation} \label{conditional_H^3_eq}
\rho_{c_t} (c| w_t=\cosh(\lambda_t)) = \frac{1}{{\rm Vol}(\mathcal{C}_\lambda)} {\rm DH}_\lambda.
\end{equation}
\end{Thm}

\begin{proof} [Proof of Theorem \ref{conditional_H^3}]
We use the same method as for Theorem \ref{conditional_S^3}, comparing two Fokker-Plank equations. 

The first Fokker-Planck equation is obtained from the first equation of the system \eqref{system_H^3}, and it is describing the evolution of the probability density of $w_t$:
\begin{equation} \label{FP1_H}
\frac{d}{d t} p_w 
= \frac{1}{2} (w^2-1) \frac{\d^2}{\d w^2} p_w 
+ \frac{1}{2} w \frac{\d}{\d w} p_w - \frac{1}{2} p_w,
\end{equation}
where $p_w = p_w(w,t)$.

The Fokker-Planck equation describing the process $(w_t, c_t)$  is derived from \eqref{system_H^3}, and it is as follows:
\begin{equation} \label{FP2_H}
\begin{array}{lll}
\frac{d}{d t} p_{w,c} 
&=& 
\frac 1 2 (w^2 - 1) \frac{\d^2}{\d w^2} p_{w,c} + 
 \frac{1}{2} c^2 \frac{\d^2}{\d c^2} p_{w,c} + (cw - 1) \frac{\d^2}{\d w \d c} p_{w,c} \\
&& + \frac 3 2 w \frac{\d}{\d w} p_{w,c}  + \frac 3 2 c \frac{\d}{\d c} p_{w,c}. 
\end{array}
\end{equation}

If $p_{w,c}$ is of the form 
\begin{equation} \label{FPanz_H}
p_{w,c}(w,c,t) = p_w(w,t) \times \frac{1}{2 \sqrt{w^2-1}} \chi_{[w - \sqrt{w^2 - 1},w + \sqrt{w^2 - 1}]}(c),
\end{equation}
implied by \eqref{conditional_H^3_eq}, 
then \eqref{FP1_H} and \eqref{FP2_H} coincide. 
The existence and the uniqueness of solution of \eqref{FP1_H} and \eqref{FP2_H} is guaranteed as they are the projections of Brownian motion. The equation \eqref{FP2_H} together with initial conditions 
\begin{equation*}
p_{0}(x,y) = \delta_{\sqrt{1+ \eps^2}}(x) \times \tfrac{1}{2 \eps} \chi_{[\sqrt{1+ \eps^2}-\eps,\sqrt{1+ \eps^2}+\eps]}(y),    
\end{equation*}
gives rise to a unique solution, and it is of the form \eqref{FPanz_H}. Thus, the solution is of this form for $\eps \to 0$ as well, and this finishes the proof.
\end{proof}

\Addresses

\end{document}